\documentclass[11pt, twoside]{article}
\usepackage[margin=1in]{geometry}

\usepackage[utf8]{inputenc} % allow utf-8 input
\usepackage[T1]{fontenc}    % use 8-bit T1 fonts
\usepackage{hyperref}       % hyperlinks
\usepackage{booktabs}       % professional-quality tables
\usepackage{amsfonts}       % blackboard math symbols
\usepackage{nicefrac}       % compact symbols for 1/2, etc.
\usepackage{microtype}      % microtypography
\usepackage{color}
\usepackage{mathtools}
\usepackage{amsthm,amsmath,amssymb,graphicx,url}%  
\usepackage{epstopdf}
\usepackage{enumitem}
\setlist{nosep,after=\vspace{\baselineskip}}
\usepackage{algorithm,algpseudocode}
\DeclareMathOperator*{\argmax}{arg\,max}

\usepackage{dsfont}
\usepackage{authblk}

\newtheorem{definition}{Definition}
\newtheorem{lemma}{Lemma}

\newtheorem{theorem}{Theorem}

\def\ci{\perp\!\!\!\perp}

\def\NP{\textsc{NP}}
\newcommand{\mat}[1]{\mathbf{#1}}

\usepackage{setspace}	
\begin{document}

\title{Entropic Causality and \\ Greedy Minimum Entropy Coupling}
\author[1,*]{Murat Kocaoglu}
\author[1,\textdagger]{Alexandros G. Dimakis}
\author[1,\textdaggerdbl]{Sriram Vishwanath}
\author[2,\textsection]{Babak Hassibi}
\affil[1]{\small Department of Electrical and Computer Engineering, The University of Texas at Austin, USA}
\affil[2]{\small Department of Electrical Engineering, California Institute of Technology, USA}
\affil[ ]{\small \textit \textsuperscript{*} mkocaoglu@utexas.edu \textsuperscript{\textdagger}dimakis@austin.utexas.edu  \textsuperscript{\textdaggerdbl}sriram@ece.utexas.edu \textsuperscript{\textsection}hassibi@systems.caltech.edu}
\renewcommand\Authands{ and }

\maketitle
\begin{abstract}
We study the problem of identifying the causal relationship between two discrete random variables
from observational data. 
We recently proposed a novel framework called \textit{entropic causality} that 
works in a very general functional model but makes the assumption that the unobserved exogenous variable 
has small entropy in the true causal direction. 

This framework requires the solution of a \textit{minimum entropy coupling problem:} 
Given marginal distributions of $m$ discrete random variables, each on $n$ states, find the joint distribution 
with minimum entropy, that respects the given marginals. 
This corresponds to minimizing a concave function of $n^m$ variables
over a convex polytope defined by $n\, m$ linear constraints, called a transportation polytope. 
Unfortunately, it was recently shown that this minimum entropy coupling problem is \NP-hard, even for 2 variables with $n$ states. Even representing points (joint distributions) over this space can require exponential complexity (in $n,m$) if done naively. 

In our recent work we introduced an efficient greedy algorithm to find an approximate solution for this problem. 
In this paper we analyze this algorithm and establish two results: that our algorithm always finds a local minimum and also is within an additive approximation error from the unknown global optimum.
\end{abstract}

\section{Introduction}
Causality is of interest to statisticians, philosophers, engineers and medical scientists \cite{Chalupka2017,Wentrup2016, Ramsey2010}. Understanding the causal relations between observable parameters is important in analyzing the workings of a system, as well as predicting how it will behave after a policy change. 
Causality has been studied under several frameworks including potential outcomes~\cite{Rubin1974} and structural equation modeling~\cite{Pearl2009}. In this paper we rely on structure equation models and data-driven causality 
using information theory. 

The use of information theoretic tools for causal discovery is recently gaining increasing attention
through various approaches: For example, Janzing et al.~\cite{Janzing2012} propose an information geometry approach that relies on a cause and mechanism independence assumption. Another line of work focuses on 
time-series data and uses Granger causality and 
directed information~\cite{Granger1969,Etesami2016,Quinn2015,Kontoyiannis2016}.
In this paper we also use information measures but rely on a different framework 
that we recently proposed~\cite{Kocaoglu2017}.

Our framework, called entropic causality~\cite{Kocaoglu2017} is data-driven, i.e., it can estimate causal directions between two discrete random variables without interventions. Our approach uses R\'{e}nyi entropy as a complexity 
measure and considers \textit{the simpler model more likely to be the true causal direction}. 
In~\cite{Kocaoglu2017} we showed that finding the simplest causal model that explains an observed joint distribution requires solving a \textit{minimum entropy coupling problem}: Given marginal distributions of $m$ discrete random variables, each on $n$ states, find the joint distribution with minimum entropy, that respects the given marginals. This corresponds to minimizing a concave function of $n^m$ variables over a convex polytope defined by $n\, m$ linear constraints, called a transportation polytope \cite{Loera2013}. 

The minimum entropy coupling problem between two variables was shown to be NP-hard in \cite{Kovacevic2012}. In \cite{Kocaoglu2017}, we proposed a greedy algorithm for the minimum entropy coupling problem and showed that for two variables, it always finds a local optimum. The proof used a characterization of the KKT conditions of the corresponding optimization problem and a characterization of the algorithm output when there are two variables. 
However, this characterization cannot be used when there are more variables. 

In this work, we extend the result in \cite{Kocaoglu2017}: We develop a new characterization of the algorithm output for any number of variables. This characterization allows us to conclude that the algorithm output satisfies the KKT conditions irrespective of the number of variables, which implies that the algorithm returns a local optimum. Moreover, we show an additive approximation guarantee with respect to the global optimum. 

In Section \ref{sec:relatedwork}, we provide a very short overview of the causal inference literature. In Section \ref{sec:Background}, we summarize the results of \cite{Kocaoglu2017} and explain how minimum entropy coupling  arises in the entropic causal inference framework. In Section \ref{sec:LocalOptimum}, we identify the conditions necessary for a solution to be a local optimum and show that our algorithm's output always satisfies these conditions by deriving a new characterization. In Section \ref{sec:Approximation}, we develop our approximation guarantee for a variant of this algorithm, which is easier to analyze. 

\section{ Related Work}
\label{sec:relatedwork}

Causal relationships between random variables can be represented by causal directed graphical models~\cite{Pearl2009, Spirtes2001}. Pearl's framework led to a complete graph theoretic characterization of which parts of a causal graph are learnable using statistical tests. Efficient algorithms were developed for this learning task by Spirtes et al. \cite{Spirtes2001}. 
Unfortunately, a general causal graph cannot be uniquely identified from data samples.

A complete solution to the causal graph identification problem requires experiments, also called interventions.  An intervention forces the value of a variable without affecting the other system variables. This removes the effect of its causes, effectively creating a new causal graph. These changes in the causal graph create a post-interventional distribution among variables, which can be used to learn additional causal relations in the original graph. The procedure can be applied repeatedly to fully identify any causal graph \cite{Shanmugam2015}. There is significant progress recently on how to efficiently perform experiments \cite{EberhardtThesis,Shanmugam2015}, even under constraints \cite{Kocaoglu2016}. Unfortunately, in many cases it is very difficult (or even impossible) to perform experiments and we are only given a static dataset.
%Even the case of learning the causal relation between two variables from data is challenging\footnote{One can argue that it is especially challenging since statistical tests cannot be used unlike when 3 or more variables are present.}. 

When performing experiments is not an option, to identify the causal relations between the variables we need additional assumptions on the data generating process. The most widely employed assumption is the additive noise assumption, which asserts that the unobserved variables affect the observable variables additively. Under this assumption, authors in \cite{Hoyer2008} showed that, except for a measure zero parameter set, one can identify the true causal direction between two variables, as long as the relation is non-linear. A similar result is known when the noise is non-Gaussian, irrespective of the relation between the variables \cite{Shimizu2006}. These approaches inherently assume continuous variables and additive noise. Other works consider discrete variables with the additive noise \cite{Peters2011}, or continuous variables without the additive noise assumption \cite{Mooij2010}. 

Another approach is to exploit the postulate that the cause and mechanism are in general independently assigned by nature. The notion of \emph{independence} here is captured by assigning maps, or conditional distributions to random variables to argue about independence of cause and mechanism. In this direction an information-geometry based approach is suggested \cite{Janzing2012}. Independence of cause and mechanism is captured by treating the log-slope of the function as a random variable, and assuming that it is independent from the cause. In the case of a deterministic relation $Y=f(X)$, there are theoretical guarantees on identifiability. However, this assumption is restrictive for real data.

In \cite{Kocaoglu2017}, we introduced the entropic causality framework. Our framework does not assume additive noise
and uses probability distributions as opposed to variable values. Thus, it can naturally handle
both categorical as well as ordinal variables. 
The central postulate is that in the true direction, the R\'{e}nyi entropy of the exogenous variable is small. 
The central theoretical result of~\cite{Kocaoglu2017} is identifiability for zero order R\'{e}nyi entropy (i.e., support of distribution): If the cardinality of the exogenous variable is small in the true direction, then there does not exist any causal model where the cardinality of the exogenous variable in the reverse direction is also small, under mild assumptions. We conjecture that a similar identifiability result is true for R\'{e}nyi entropy of order 1, i.e., Shannon entropy, and numerical simulations seem to verify it. 
Furthermore, we showed that the corresponding causality test can match or outperform the previous state of the art in causal identification benchmarks in real and synthetic datasets~\cite{Kocaoglu2017}.

In very recent parallel work, Cicalese et al.~\cite{Cicalese2017} proposed a more involved greedy algorithm for the minimum entropy coupling problem and showed a very strong 1-bit approximation guarantee for it. The proposed algorithm only applies for two variables.
Two variable algorithms for minimum entropy coupling
can only be used for entropic causality if one of the two variables takes only two-values. 
Therefore, it would be very interesting if it can be extended for multiple variables,
especially if similar strong approximation guarantees are true. 
 
\section{Background}
\label{sec:Background}
\subsection{Notation}
We use uppercase letters ($X$) for random variables, lowercase letters for their realizations and constants ($x,i,\alpha$), lowercase bold letters for column vectors ($\mat{p}$), uppercase bold letters for matrices and tensors $(\mat{G})$. We represent the set $\{1,2,\hdots,n\}$ by $[n]$, whereas $[a,b]$ indicates the continuous interval from $a$ to $b$ as usual. Vectors and sets with indices are simply represented through subscripts as follows: $[x_i]_{i\in[n]}$ represents the column vector $[x_1,x_2,\hdots,x_n]^T$ and $\{u_i\}_{i\in[m]}$ represents the set $\{u_1,u_2,\hdots,u_m\}$. $X\sim p_X$ means the random variable $X$ is distributed with the probability mass function $p_X$, i.e., $\Pr(X=i) = p_X(i)$. $\ci$ stands for the statistical independence between random variables. The Shannon entropy $H([p_i]_i) = -\sum_{i}p_i\log(p_i)$ naturally extends to matrices (and tensors) as $H([r_{i,j}]_{i,j}) = -\sum_{i,j} r_{i,j}\log(r_{i,j})$, where $\log(.)$ stands for the logarithm base 2.

\subsection{Causal Model}
In this section, we introduce Pearl's causal model for two variables and no unobserved common causes.
Causal models are powerful because they can answer hypothetical questions involving experiments. An experiment, called an intervention in this context, means forcing a set of random variables to take certain values. This operation is captured by the \emph{do(.)} operator of Pearl \cite{Pearl2009}. Thus, by definition, the causal model captures the knowledge of what will happen after performing any intervention on the observed variables. Consider two variables $X,Y$. Suppose $X$ causes $Y$. The following are what this causal model entails: (i) There exists an exogenous (unobserved) random variable $E\ci X$ and a map $f$ such that $Y=f(X,E)$. Let $E\sim p_E,X\sim p_X$. (ii) An intervention $do(X=x)$ changes the data generating model and yields $X=x,E\sim p_E,Y=f(x,E)$. Thus, an intervention on $X$ does not change the distribution of $E$, but fixes the value of $X$. Hence the distribution of $Y$ is affected through these changes. However, an intervention on $Y$ has a different effect. (iii) $do(Y=y)$ changes the model as follows: $X\sim p_X,E\sim p_E, Y=y$. The important thing to notice here is that intervening on $Y$ makes it independent from $X$, whereas intervening on $X$ does not make it independent from $Y$.\footnote{Technically, to talk about statistical independence, we need stochastic interventions: Consider $do(X=U)$ which forces $X$ to take the same values as an independent random variable $U$.}

The fact that a causal model can answer interventional queries is what makes it so powerful, but also hard to learn from data. In general, given a joint distribution over $X,Y$ one can find functions $f,g$ where $Y=f(X,E),E\ci X$ and $X=g(Y,\tilde{E}), \tilde{E}\ci Y$. This makes the problem of learning the causal relation between $X$ and $Y$ unidentifiable in general. The objective of data driven causal inference is to identify the assumptions on either the function $f$ or the variable $E$, under which the causal model can be learned. 

\subsection{The Entropic Causal Inference Framework} 
Entropic causal inference \cite{Kocaoglu2017} uses the number of random bits as a complexity measure and chooses the simpler model as the true causal model. Suppose we observe the joint distribution of two variables $X,Y$ each with $n$ states. Consider the problem of identifying the exogenous variable with minimum Shannon entropy such that there is a causal model where $X$ causes $Y$, that yields this joint distribution. In \cite{Kocaoglu2017}, we established that this problem is equivalent to the minimum entropy coupling problem between $n$ variables each with $n$ states. 

Consider the variables $X,Y$ with $X,Y\in[n]$. Suppose $X$ causes $Y$. Then $Y=f(X,E)$, where $E$ is an exogenous variable of cardinality $m$ for some $m$ independent from $X$, and $f$ is some map $f:[n]\times [m]\rightarrow [n]$. Let $U_i$ be a random variable that has the same distribution as the distribution of $X$ conditioned on $Y=i$: $\Pr(U_i=j)=\Pr(X=j|Y=i)$. We have the following lemma:

\begin{lemma}\cite{Kocaoglu2017}
\label{lem:equivalence}
Let $X,Y$ be two variables with $X,Y\in[n]$. Consider any causal model $X=g(Y,\tilde{E}),\tilde{E}\ci Y$. Then $H(\tilde{E})\geq H^*(U_1,\hdots, U_n)$, where $H^*(U_1,\hdots, U_n)$ is the minimum joint entropy of variables $\{U_1,\hdots,U_n\}$ subject to the constraint that each $U_i$ has the same marginal distribution as the conditional distribution of $X$ given $Y=i$.

\noindent
Moreover, there is an $\tilde{E}\ci Y$ with $H(\tilde{E}) = H^*(U_1,\hdots, U_n)$.
\end{lemma}
\begin{proof}
See the proof of Theorem 3 in the appendix of \cite{Kocaoglu2017}.
\end{proof}

Lemma \ref{lem:equivalence} puts the minimum entropy coupling problem at the center of the entropic causal inference framework. If we could solve the minimum entropy coupling problem, we could identify the exogenous variable with minimum entropy. If the identifiability result holds (Conjecture 1 in \cite{Kocaoglu2017}), $H(Y)+H(\tilde{E})$ will be greater than $H(X)+H(E)$ if entropy of $E$ is sufficiently small. Hence, closely approximating the minimum entropy coupling is essential for an effective causal inference algorithm using the entropic causal inference framework. 
\subsection{Greedy Minimum Entropy Coupling Algorithm}
Different from \cite{Kocaoglu2017}, we provide the version of the greedy minimum entropy coupling algorithm that constructs the joint distribution tensor, rather than only the non-zero probability values, which is more instructional for this paper. The greedy algorithm is given in Algorithm \ref{alg:greedy}. The marginal distribution of variable $i$ is shown by the column vector $\mat{p_i}$. Note that in practice, one would only store the non-zero probability values output by the algorithm, rather than creating the extremely sparse tensor $\mat{P}$ with $n^m$ entries.

\begin{algorithm}[ht!]
\begin{small}
    \caption{Joint Entropy Minimization Algorithm}
   \label{alg:greedy}
\begin{algorithmic}[1]
    \State {\bfseries Input:} Marginal distributions of $m$ variables each with $n$ states $\{\mat{p_1},\mat{p_2},...,\mat{p_m}\}$.
    \State Initialize the tensor $\mat{P}(i_1,i_2,\hdots,i_n) = 0, \forall i_j\in [n],\forall j\in[n]$.%[\mat{0}]^n,$ where $\mat{0}=[0, 0, \hdots,0]\in \mathbb{R}^n$.
    \State Initialize $r=1$.
	\While  {$r>0$} 
	\State $(\{\mat{p_i}\}_{i\in[m]}, r) = \textbf{UpdateRoutine}(\{\mat{p_i}\}_{i\in[m]}, r)$
    \EndWhile
    \State \Return $\mat{P}$.
	\State{\bfseries UpdateRoutine($\{\mat{p_1},\mat{p_2},...,\mat{p_m}\}, r$)}
    \State Find $i_j\coloneqq \argmax_{k}\{\mat{p_j}(k)\},\forall j\in[m]$.
    \State Find $u=\min\{\mat{p_k}(i_j)\}_{k\in[n]}$.
    \State Assign $\mat{P}(i_1,i_2,\hdots,i_n) = u$.
    \State Update $\mat{p_k}(i_j)\leftarrow \mat{p_k}(i_j)-u, \forall k\in[m]$.
    \State Update $r=\sum_{k\in [n]}{\mat{p_1}(k)}$
	\State \Return $\{\mat{p_1},\mat{p_2},...,\mat{p_m}\}, r$
\end{algorithmic}
\end{small}
\end{algorithm}

At each iteration, the algorithm finds the largest probability mass in each marginal, and assigns the minimum of these to the corresponding coordinate in the joint probability tensor. The motivation is that, the large chunks of probability masses are not split into smaller chunks, making as small contribution as possible to the total entropy. The algorithm satisfies at least one marginal constraint at each step, and $m$ of them in the last step. Thus it terminates in at most $nm-m+1$ steps.

\section{Greedy Algorithm Gives Local Optimum}
\label{sec:LocalOptimum}
In this section, we present our main theorem and show that the greedy algorithm always finds a local optimum. We consider $n$ variables each with $n$ states. The extension of the analysis to $m$ variables each with $n$ states is trivial. Let us first formalize the entropy minimization problem:

\begin{definition}[Minimum Entropy Coupling]
Let $U_i, i\in [n]$ be discrete random variables with $n$ states, with marginal distributions $\mat{p_i}\in [0,1]^n$. The minimum entropy coupling problem is to find the joint distribution with minimum entropy that is consistent with the given marginals:
\begin{align}
 &\underset{p(U_1,U_2,\hdots,U_n) }{\min} \hspace{0.2in} H(U_1,U_2,\hdots,U_n) \nonumber\\
 &\emph{s. t.}  \hspace{0.1in} \sum_{j\neq i}{\sum_{u_j\in[n] }p(u_1,u_2,\hdots,u_n)} = \mat{p_i}(u_i), \forall i,u_i.
\end{align}
\end{definition}
\hspace{0.1in}
We can equivalently write down this optimization problem by representing the joint probability value for each configuration as a different variable. This representation has $n^n$ variables and $n^2$ constraints ($n$ marginals and $n$ points for each marginal). Let $x(i_1,i_2,\hdots,i_n)$ be a variable for every $n$-tuple $(i_1,i_2,\hdots,i_n)\in [n]^n$. Notice that the index for $j$th dimension, i.e., $i_j$, captures the realization of variable $U_j$. Then the optimization problem can be written as follows:
\begin{align}
\label{eq:minentropyoptimization}
 & \underset{x }{\text{min}}\hspace{0.16in} 
  \sum_{i_j\in [n], \forall j\in[n]} -x(i_1,i_2, \hdots, i_n)\log{x(i_1,i_2, \hdots, i_n)} \nonumber\\
 & \text{s. t.}\hspace{0.1in} 
   \sum_{i_k\in[n],\forall k\neq l, i_l = j}{x(i_1,i_2, \hdots, i_n)} = p_l(j), \forall j,l\in[n] \nonumber\\
&\hspace{0.3in}   x(i_1,i_2, \hdots, i_n)\geq 0, \forall (i_1,i_2,\hdots,i_n)\in[n]^n \hspace{0.1in}  
\end{align}
In (\ref{eq:minentropyoptimization}), we dropped the constraint $\sum_{j,i_j} x(i_1,i_2 \hdots, i_n)=1$. Total sum is equivalent to first marginalizing out dimensions $1$ to $n-1$, and then marginalizing out dimension $n$. If marginalizing out the first $n-1$ dimensions gives $\mat{p_n}$, which is already captured as a separate equality constraint, summing across this dimension gives 1 since $\mat{p_n}$ sums to 1.

In this section, we show the following theorem:
\begin{theorem}
\label{thm:greedyLocalOpt}
Algorithm \ref{alg:greedy} finds a local optimum point of the optimization problem in (\ref{eq:minentropyoptimization}).
\end{theorem}	
\subsection{KKT Conditions}
First, we characterize the points that satisfy the KKT conditions. We have the following lemma:
\begin{lemma}
\label{lem:KKT}
Consider the optimization problem in (\ref{eq:minentropyoptimization}). Let $x^*(i_1,i_2, \hdots, i_n), i_j\in[n],j\in[n]$ be a point that satisfies the KKT conditions. Then there are $n$ vectors $\mat{u_k},k\in[n]$ each of length $n$ such that either $x^*(i_1,i_2, \hdots, i_n)=0$, or
\begin{equation}
\label{eq:optCharacterize}
 \log{x^*(i_1,i_2, \hdots, i_n)}+1 = \sum\nolimits_{k\in[n]}\mat{u_k}(i_k).
\end{equation}
\end{lemma}
\begin{proof}
Consider the following general optimization problem:
\begin{equation}
\label{eq:generaloptimization}
\begin{aligned}
 & \underset{x }{\text{min}}
 & &f_0(x)\\
 & \text{s. t.}
 & & h_i(x) = 0, i\in[p] \\
 & &  & f_i(x)\leq 0, i\in [m],
\end{aligned}
\end{equation}
Lagrangian becomes 
\begin{equation}
L(x,\lambda,v) = f_0(x)+\sum_{i=1}^m\lambda_i f_i(x)+\sum_{i=1}^pv_ih_i(x),
\end{equation}
which gives the KKT conditions
\begin{align*}
&f_i(x^*)\leq 0, i\in [m]\\
&h_i(x^*) = 0, i\in [p]\\
&\lambda_i^*\geq 0, i\in[m]\\
&\lambda_i^*f_i(x^*) = 0, i\in [m]\\
&\nabla L(x^*,\lambda^*,v^*) = 0
\end{align*}
This implies, for fixed $i$, either $f_i(x^*) = 0$ or $\lambda_i^*=0$. Matching the constraints in (\ref{eq:minentropyoptimization}) to the functions in (\ref{eq:generaloptimization}), we identify $f(i_1,i_2, \hdots, i_n)$ and $h_{l,j}$ as follows:
\begin{align*}
&f(i_1,i_2, \hdots, i_n) = -x(i_1,i_2, \hdots, i_n), \forall (i_1,i_2, \hdots, i_n)\in[n]^n\\
&h_{l,j}  = \sum_{i_k\in[n] \forall k\neq l, i_l = j}x(i_1,i_2, \hdots, i_n) - p_l(j), \forall l\in[n],j\in[n]
\end{align*}

The Lagrangian of (\ref{eq:minentropyoptimization}) can be written as follows:
\begin{align}
L(x,\lambda, \nu) &= \sum_{\substack{\mathclap{i_j\in [n], \forall j\in[n]}}} -x(i_1,i_2, \hdots, i_n)\log{x(i_1,i_2, \hdots, i_n)}\nonumber\\
&- \sum_{\substack{\mathclap{(i_1,i_2, \hdots, i_n)\in[n]^n}}} \lambda(i_1,i_2, \hdots, i_n)x(i_1,i_2, \hdots, i_n) \nonumber\\
&+ \sum_{j\in [n],l\in[n]}\nu_{l,j}(\sum_{\mathclap{\substack{\hspace{0.3in}i_k\in[n]\forall k\neq l,\\ i_l = j}}}x(i_1,i_2, \hdots, i_n) - p_l(j)), 
\end{align}
for the dual parameters $\lambda(i_1,i_2, \hdots, i_n)$ and $\nu_{l,j}$. The gradient being zero gives us the following:
\begin{align*}
\frac{\partial L}{\partial x(i_1,i_2, \hdots, i_n)} &= -\log{x(i_1,i_2, \hdots, i_n)} - 1  \nonumber\\
&- \lambda(i_1,i_2, \hdots, i_n) + \sum_{l\in[n]}\hspace{0.2in} \nu_{l,i_l}=0\\
\frac{\partial L}{\partial \lambda(i_1,i_2, \hdots, i_n)}& =  \left\{
  \begin{array}{@{}ll@{}}
    0, & \text{if}\ \lambda(i_1, \hdots, i_n)=0 \\
    -x(i_1, \hdots, i_n), & \text{if}\ \lambda(i_1, \hdots, i_n)\neq 0
  \end{array}\right.
\end{align*}

The conditions above imply the following for the optimal point $x^*$: Either $x^*(i_1,i_2, \hdots, i_n)=0$ or if $x^*(i_1,i_2, \hdots, i_n)\neq 0$ it satisfies
\begin{equation}
x^*(i_1,i_2, \hdots, i_n) = 2^{-1+\sum_{k\in[n]}\nu_{k,i_k}}.
\end{equation}
Thus, for $n$ vectors $u_k\coloneqq \nu_{k,.}, k\in[n]$ of length $n$, we have $\log{x^*(i_1,i_2, \hdots, i_n)}+1 = \sum_{k\in[n]}u_k(i_k).$
\end{proof}

By Lemma \ref{lem:KKT}, the optimal point satisfies the following: Each nonzero joint probability can be written as a product of the corresponding entries of $n$ vectors $\{\mat{v_k}\}_{k\in [n]}$ of length $n$. Inspired by the definition of independence, we will term such joint distributions as \emph{quasi-independent}:
\begin{definition}
A joint distribution $p(X_1,X_2,\hdots,X_m)$ for $X_i\in[n]$ is called quasi-independent, if there are $m$ vectors $\mat{u_j},j\in[m]$ such that either $p(i_1,i_2,\hdots,i_m) = 0$ or $p(i_1,i_2,\hdots,i_m) = \prod_{j\in[m]}\mat{u_j}(i_j), \forall i_j\in[n],j\in[m]$.
\end{definition}

\subsection{Characterization of Greedy Algorithm Output}
Consider Algorithm \ref{alg:greedy}. It selects the minimum of maximum probability values across each marginal at each step, subtracts this probability mass from the corresponding coordinates in each marginal and iterates. Next, we show that one can always construct $\mat{u_k}$ vectors that satisfy $\log{x(i_1,i_2, \hdots, i_n)}+1 = \sum_{k\in[n]}\mat{u_k}(i_k)$, where $x(i_1,i_2, \hdots, i_n)$ is the probability mass assigned to point $(i_1,i_2, \hdots, i_n)$ by the algorithm.

Let the algorithm select a probability mass for the point $S_j = (i_1^j,i_2^j, \hdots, i_n^j)$ at iteration $j$. $x(S_j)> 0$. Let $a_j\coloneqq \log{x(S_j)}+1$ after this assignment. Define the column vector $\mat{u} \coloneqq [\mat{u_1}^T,\hdots,\mat{u_n}^T]^T$. $\{\mat{u_i}\}_{i\in[n]}$ are length-$n$ vectors to be decided. We will show that, given the assignments made by the algorithm, one can always construct a $\mat{u}$ such that (\ref{eq:optCharacterize}) holds.

Observe that each iteration of the algorithm corresponds to a linear equation in $\mat{u}$. Note that $\mat{u}$ has length $n^2$ and at iteration $j$, $\mat{u}$ should satisfy the constraint $\mathds{1}_{S_j}^T \mat{u}=a_j$,
where $\mathds{1}_{S_j}^T$ is the indicator vector that is 1 in the columns from $S_j$ and zero otherwise: If $S_j = (i_1^j,i_2^j, \hdots, i_n^j)$, then $\mathds{1}_{S_j}(k)=1, \forall k\in \xi_j$, where $\xi_j= \{i_t^j+(t-1)n\}_{t\in [n]}$. %Notice that each block of length $n$ contains exactly a single 1. 
We know that the algorithm terminates in at most $n(n-1)+1$ steps. Thus, we have $m<n(n-1)+1$ linear equations and $n^2$ variables. This corresponds to a system of linear equations $\mat{Gu} = \mat{a}$, where $\mat{G}(j,:) = \mathds{1}_{S_j}^T$ and $\mat{a} = [a_j]_{j\in[m]}$ is a column vector.

We have the following key observation: At each iteration step, the algorithm satisfies at least one of the marginal constraints, since it chooses the \textit{minimum} of maximum probabilities. Thus, if at iteration $j$ the algorithm select the set of the coordinates $(i_1^j,i_2^j, \hdots, i_n^j)$, then for some $k\in[n]$ algorithm  never selects the coordinate $i_k^j$ again, since the corresponding marginal constraint is already satisfied. In terms of the matrix $\mat{G}$, this translates to the following statement: \textit{Every row $j$ of $\mat{G}$ contains a column $k\in \xi_j$ where $\mat{G}(l,k) = 0, \forall l>j$}. Thus, every row of $\mat{G}$ has a column where that row contains the last 1 in that column. We have the following lemma:
\begin{lemma}
\label{lem:helper}
Let $\mat{G}$ be a ${0,1}$ matrix where no row is identically zero. If for every row $j$, of all the columns with value 1, there exists a column $k$ such that $\mat{G}(l,k) = 0,\forall l>j$, then the rows of $\mat{G}$ are linearly independent.
\end{lemma}
\begin{proof}
Assume otherwise. Then there exists a set of rows $S$ and coefficients $\alpha_j>0$ such that $\sum_{j\in S}\alpha_j\mat{G}(j,:)=0$. Let $l = \min\{i:i\in S\}$. By definition, $l^{th}$ row of $\mat{G}$ has a column $k$ with $\mat{G}(t,k)=0,\forall
t>l$. Thus, this column cannot be made 0 using a linear combination of rows with a larger index, which contradicts with $\sum_{j\in S}\alpha_j\mat{G}(j,:)=0$.
\end{proof}

By Lemma (\ref{lem:helper}), the rows of $\mat{G}$ are linearly independent.  This is also true for the augmented matrix of the system $\mat{Gu} = \mat{a}$. Hence, the assignments are consistent and there is at least one solution to the linear system $\mat{Gu} = \mat{a}$.

\begin{proof}[Proof of Theorem \ref{thm:greedyLocalOpt}]
Consider the joint distribution output by the greedy algorithm. From the above discussion, the assignments to the joint distribution by the greedy entropy minimization algorithm can always be used to create $n$ vectors, such that the points where the joint is non-zero can be written as the product of the corresponding coordinates of these $n$ vectors. Thus, the greedy algorithm outputs a point which is quasi-independent, and satisfies the KKT conditions of the minimum entropy coupling problem. Hence, this is a stationary point. Since entropy is a concave function, there are no saddle points. Thus, greedy algorithm outputs a local optimum.
\end{proof}

\section{Approximation Guarantee}
\label{sec:Approximation}
In this section, we analyze a variant of the greedy algorithm, Algorithm \ref{alg:alternative}, which is easier to develop an approximation guarantee for.
\begin{algorithm}[ht!]
\begin{small}
    \caption{Joint Entropy Minimization - Alternative}
   \label{alg:alternative}
\begin{algorithmic}[1]
    \State {\bfseries Input:} Marginal distributions $\{\mat{p_1},\mat{p_2},...,\mat{p_m}\}$.
    \State Initialize the tensor $\mat{P}(i_1,i_2,\hdots,i_m) = 0, \forall i_j\in [n],\forall j\in[m]$.
    \State Initialize empty sets $S_j=\emptyset, \forall j\in[m]$.
    \State \textbf{Phase I}
    \For{$1\leq t\leq n$}
    \State Find $i_j^t\coloneqq \argmax_{k\in [n]\backslash S_j}\{\mat{p_j}(k)\},\forall j\in[n]$.
    \State Find $p_{min}(t)=\min\{\mat{p_k}(i_j^t)\}_{k\in[n]}$.
    \State Assign $\mat{P}(i_1^t,i_2^t,\hdots,i_n^t) = p_{min}(t)$.
    \State Update $\mat{p_k}(i_j^t)\leftarrow \mat{p_k}(i_j^t)-p_{min}(t)$.
	\State Update $S_j\leftarrow S_j \cup \{i_j^t\}$.
	\EndFor
	\State Initialize $r=\sum_{k\in [n]}{\mat{p_1}(k)}$
	\State \textbf{Phase II}	
	\While  {$r>0$} 
	\State $(\{\mat{p_i}\}_{i\in[m]}, r) = \textbf{UpdateRoutine}(\{\mat{p_i}\}_{i\in[m]}, r)$
    \EndWhile
    \State \Return $\mat{P}$.
\end{algorithmic}
\end{small}
\end{algorithm}
Different from Algorithm \ref{alg:greedy}, Algorithm \ref{alg:alternative} looks at each value of every given marginal exactly once during Phase I. This allows us to relate the entropy contribution of Phase I to a lower bound to the optimum entropy.

Consider two random variables $X_1,X_2$. We use $\mu_1,\mu_2$ to represent the marginal distributions of $X_1$ and $X_2$ after sorting their probabilities in decreasing order. We can extend the entropy function to operate on vectors which do not necessarily sum to 1. To make the distinction from entropy, we use $h(.)$ for this operator\footnote{$h(.)$ is often used for the differential entropy operator. Since we do not use differential entropy in this paper, we believe this is not a source of confusion.}.
\begin{theorem}
\label{thm:approx}
Let $X_1,X_2$ be two discrete random variables with $n$ states and $\mu_1=[p_1(i)]_{i\in[n]}$, $\mu_2=[p_2(i)]_{i\in[n]}$ be their marginal distribution vectors sorted in decreasing order. Let $p_m(i) = \min\{p_1(i),p_2(i)\}$. Let $U$ be the joint distribution output by the greedy algorithm, and $H^*(X_1,X_2)$ the minimum joint entropy of all joints that respect the marginals. Then
\begin{equation*}
H(U)\leq H^*(X_1,X_2)+1-T\log(1/T) + \min\{h(l_1),h(l_2)\},
\end{equation*}
where $l_j = [p_j(i)-p_m(i)]_{i\in[n]}$ for $j\in \{1,2\}$, and $T = 0.5\sum_{i\in [n]}\lvert p_1(i)-p_2(i) \rvert$ is the total variation distance between the sorted marginals of $X_1$ and $X_2$.
\end{theorem}
\begin{proof}
Define $p_{m}(i) = \min\{p_1(i),p_2(i)\}$. In Phase I, algorithm chooses $p_{m}(i)$ for $i\in[n]$. Consider
\begin{align}
H_a &= H(p_{m}(1), p_1(1)-p_{m}(1),p_{m}(2), p_1(2)-p_{m}(2),\nonumber\\
&\hdots,p_{m}(n), p_1(n)-p_{m}(n)).
\end{align} 
$H_a$ is the entropy of the distribution which is obtained by splitting $p_1(i)$ into $p_{m}(i)$ and $p_1(i)-p_{m}(i)$. Since each probability value is divided into at most 2 probability values,
\begin{equation}
\label{eq:ha}
H_a\leq H(X_1)+1.
\end{equation}
Similarly, we can write
\begin{align}
\label{eq:hb}
H_b&\coloneqq H(p_{m}(1), p_2(1)-p_{m}(1),p_{m}(2), p_2(2)-p_{m}(2),\nonumber\\
&\hdots,p_{m}(n), p_2(n)-p_{m}(n)) \leq H(X_2) +1.
\end{align}

Then in Phase I, algorithm creates an entropy contribution $H_{Ph1} = h(p_m(1),p_m(2),\hdots,p_m(n)) = -\sum_{i\in[n]}p_m(i)\log(p_m(i))$.  Based on the definitions of $l_1,l_2$
\begin{equation}
H_a = H_{Ph1} + h(l_1), \hspace{0.3in} H_b = H_{Ph1} + h(l_2).
\end{equation}

Let $\alpha\in\{0,1\}$. Combining with (\ref{eq:ha}) and (\ref{eq:hb}), we get
\begin{equation*}
H_{Ph1} + \alpha h(l_1)+(1-\alpha)h(l_2)\leq \alpha H(X_1)+(1-\alpha)H(X_2)+1.
\end{equation*}

To bound the contribution of the second phase, we use an "independence" bound. The following lemma is useful:
\begin{lemma}
\label{lem:jointNonDistribution}
Consider the vectors $\mat{p}=[p_i]_{i\in [n]}, \mat{q}=[q_i]_{i\in [n]}$ where $p_i,q_i\geq 0$ and $\sum_ip_i=\sum_iq_i=T$. Let $h(\mat{p}) = -\sum_ip_i\log(p_i)$. Let $\mat{R}(i,j)=r_{i,j}$ for $i\in[n],j\in [n]$ be a matrix with row sum equal to $\mat{p}$ and column sum equal to $\mat{q}$, i.e., $\sum_{j\in[n]}r_{i,j}=p_i$ and $\sum_{i\in[n]}r_{i,j}=q_j,\forall i,j\in [n]$. Then $h(\mat{R})\leq T\log(T)+h(\mat{p})+h(\mat{q})$. 

Moreover, when $\mat{R}$ is the outer product of $\mat{p}/\sqrt{T}$ and $\mat{q}/\sqrt{T}$, the equality holds.
\end{lemma}
\begin{proof}
Define the random variables $U$ and $V$ as the variables with marginal distributions $\mat{p}/T$ and $\mat{q}/T$, respectively. Let $\mat{S}(i,j) = [s_{i,j}]_{i,j\in[n]}$ be the joint distribution matrix for $U,V$ that respects the marginals $\mat{p}/T$ and $\mat{q}/T$. Since $H(U,V)\leq H(U)+H(V)$, we have	
\begin{align*}
-\sum_{i,j\in[n]}s_{i,j}\log(s_{i,j}) &\leq -\sum_{i\in[n]}\Big(\frac{p_i}{T}\Big)\log\Big(\frac{p_i}{T}\Big)-\sum_{j\in[n]}\Big(\frac{q_j}{T}\Big)\log{\Big(\frac{q_j}{T}\Big)}\\
&= \frac{1}{T}\left(-\sum_ip_i\log(p_i)+\sum_ip_i\log(T)-\sum_jq_j\log(q_j)+\sum_jq_j\log(T) \right)\\
&= \frac{1}{T}\Big(h(\mat{p}) + h(\mat{q})+2T\log(T)\Big)
\end{align*}
Define $\mat{R}(i,j) = r_{i,j}$ where $r_{i,j} = Ts_{i,j}$. Notice that row sum of $\mat{R}$ is $\mat{p}$ and column sum of $\mat{R}$ is $\mat{q}$. Then we have,
\begin{align*}
h(\mat{R}) &= -\sum_{i,j}r_{i,j}\log(r_{i,j}) = -\sum_{i,j}Ts_{i,j}\log(Ts_{i,j})\\
& = T\left(-\sum_{i,j}s_{i,j}\log(s_{i,j})\right)-T\log(T)\\
& = T\left(\frac{1}{T}\Big(h(\mat{p}) + h(\mat{q})+2T\log(T)\Big)\right) - T\log(T)\\
& = h(\mat{p}) + h(\mat{q}) + T\log(T)
\end{align*}
Suppose $\mat{R}(i,j) = \frac{p_iq_j}{T}$. Then we have,
\begin{align*}
h(\mat{R} ) &= -\sum_{i,j}\Big(\frac{p_iq_j}{T}\Big)\log\Big(\frac{p_iq_j}{T}\Big)\\
& = \frac{1}{T}\Big(-\sum_{i,j}p_iq_j\log(p_i)-\sum_{i,j}p_iq_j\log(q_j)+T^2\log(T)\Big)\\
& = \frac{1}{T}\Big( Th(\mat{p})+Th(\mat{q})+T^2\log(T)\Big) =h(\mat{p}) + h(\mat{q})+T\log(T). 
\end{align*}
\end{proof}

Following Lemma \ref{lem:jointNonDistribution}, the maximum contribution of the second phase to the entropy is obtained when we place the scaled outer product of the remaining probability values on the joint probability matrix. The remaining probabilities after phase 1 are $l_1$ and $l_2$ for $X_1$ and $X_2$. The remaining probability mass is the total variation distance, i.e., $\sum_il_1(i)=\sum_il_2(i)=T$. Thus, in Phase II, $l_1$ and $l_2$ contributes the entropy of  $H_{Ph2}\leq T\log{T} + h(l_1) + h(l_2)$. Finally, we can write
\begin{align}
H(U) &= H_{Ph1}+H_{Ph2} \leq  H_{Ph1}+T\log{T} + h(l_1) + h(l_2)\nonumber\\
&\leq \alpha H(X_1)+(1-\alpha)H(X_2)+1+T\log{T}+(1-\alpha)h(l_1)+\alpha h(l_2)\nonumber\\
&\leq H^*(X_1,X_2)+1-T\log(1/T)+\min\{h(l_1),h(l_2)\}.\label{eq:avgBound}
\end{align}
(\ref{eq:avgBound}) is obtained by selecting $\alpha=1$ if $h(l_1)>h(l_2)$ and $\alpha=0$ if $h(l_1)<h(l_2)$, and through the bound $H^*(X_1,X_2)\geq \max(H(X_1),H(X_2))\geq \alpha H(X_1)+(1-\alpha)H(X_2)$. 
\end{proof}

Consider the bound given in Theorem \ref{thm:approx}. $1-T\log(1/T)$ is a constant less than 1. However, the term $\min\{h(l_1),h(l_2)\}$ can scale with $\log(n)$ depending on the difference between the sorted marginals. In Section \ref{sec:example}, we give an example where $\min\{h(l_1),h(l_2)\}=\mathcal{O}(\log(n))$. Interestingly, for the same example we can show that the greedy algorithm output is at most 1 bit away from the global optimum. Thus, it may be possible to identify a tighter bound.

We can extend the analysis to the case of $m$ variables instead of only 2. We then have the following theorem:
\begin{theorem}
Let $\{X_i\}_{i\in[m]}$ be $m$ random variables each with $n$ states and $\mu_i=[p_i(j)]_{j\in[n]}, \forall i\in[m]$ be their marginal distribution vectors sorted in decreasing order. Let $p_{min}(j) = \min\{p_i(j),i\in[m]\}$. Let $U$ be the joint distribution output by Algorithm \ref{alg:alternative} and $H^*(X_1,\hdots,X_m)$ the global optimum. Then
\begin{align}
H(U)&\leq H^*(X_1,X_2,\hdots,X_m)+1-(m-1)T\log(1/T) \nonumber\\
&+ \sum_ih(l_i) - \max_i\{h(l_i)\},
\end{align}
where $l_i = [p_i(j)-p_{min}(j)]_{j\in[n]}$ for $i\in [m]$, and $T = \sum_{i\in [n]} (p_1(i)-p_{min}(i))$ .
\end{theorem}
\begin{proof}
Define $p_{min}(i) = \min_j\{p_j(i), j\in [m]\}$. In Phase 1, the algorithm chooses $p_{min}(i)$ for $i\in[n]$. Consider for all $j\in[m]$
\begin{align}
H_{a_j} &= H(p_{min}(1), p_j(1)-p_{min}(1),p_{min}(2), p_j(2)-p_{min}(2),\nonumber\\
&\hdots,p_{min}(n), p_j(n)-p_{min}(n)).
\end{align} 
$H_{a_j}$ is the entropy of the distribution which is obtained by splitting $p_j(i)$ into $p_{min}(i)$ and $p_j(i)-p_{min}(i)$. Since each probability value is divided into at most 2 probability values,
\begin{equation}
\label{eq:haMulti}
H_{a_j}\leq H(X_j)+1.
\end{equation}
In Phase I, algorithm creates an entropy contribution $H_{Ph1} = h(p_{min}(1),p_{min}(2),\hdots,p_{min}(n)) = -\sum_{i\in[n]}p_{min}(i)\log(p_{min}(i))$. Define $l_j = [p_j(1)-p_{min}(1),p_j(2)-p_{min}(2), \hdots, p_j(n)-p_{min}(n)]$ for all $j\in[m]$. Then we have
\begin{equation}
H_{a_j} = H_{Ph1} + h(l_j), \forall j\in[m].
\end{equation}
Let $\alpha_j\in[0,1]$ and $\sum_j\alpha_j=1$. Combining with (\ref{eq:haMulti}), we get
\begin{equation*}
H_{Ph1} + \sum_j\alpha_j h(l_j)\leq \sum_j\alpha_j H(X_j)+1.
\end{equation*}

To bound the contribution of the second phase, we use an "independence" bound similar to the one in the proof of Theorem \ref{thm:approx}. We need the following lemma:
\begin{lemma}
\label{lem:jointNonDistributionMultiple}
Consider the vectors $\mat{p_i}=[p_i(j)]_{j\in [n]}, {i\in [m]}$ where $p_i(j)\geq 0,\forall i\in[m],j\in[n]$ and $\sum_jp_i(j)=T,\forall i\in[m]$. Let $h(\mat{p_i}) = -\sum_jp_i(j)\log(p_i(j))$. Let $\mat{R}=[r_{i_1,i_2,\hdots,i_m}]_{i_j\in[n]}$ be a tensor that satisfies the following: $\sum_{i_k\in[n], \forall k\neq l,i_l=t}{r_{i_1,i_2,\hdots,i_m}}=p_l(t), \forall l\in[m],t\in[n]$.
Then $h(\mat{R}) \leq \sum_{i\in[m]}h(\mat{p_i}) + (m-1)T\log(T)$.

Moreover, when $\mat{R}$ is the outer product of $\frac{\mat{p_i}}{T^{\frac{m-1}{m}}}$,for all $i\in[m]$, the equality holds.
\end{lemma}
\begin{proof}
Define the random variables $U_i$ as the variables with marginal distributions $\mat{p_i}/T$ for all $i\in[n]$. Let $\mat{S}(i_1,i_2,\hdots,i_m) = [s_{i_1,i_2,\hdots,i_m}]_{i_j\in[n],\forall j\in[m]}$ be the joint distribution tensor for $U_i$ that respects the marginals $\mat{p_i}/T$ and $\mat{q_i}/T$. Since $H(U_1,U_2,\hdots,U_m)\leq \sum_iH(U_i)$, we have	
\begin{align*}
-\sum_{i_j\in[n],j\in[m]}s_{i_1,i_2,\hdots,i_m}\log(s_{i_1,i_2,\hdots,i_m}) &\leq -\sum_{i,j\in[n]}\Big(\frac{\mat{p_i}(j)}{T}\Big)\log\Big(\frac{\mat{p_i}(j)}{T}\Big)\\
&= \frac{1}{T}\left(\sum_ih(\mat{p_i})+mT\log(T)\right)
\end{align*}
Define $\mat{R}(i_1,i_2,\hdots,i_m) = r_{i_1,i_2,\hdots,i_m}$ where $r_{i_1,i_2,\hdots,i_m} = Ts_{i_1,i_2,\hdots,i_m}$. Notice that with this scaling, marginalizing out every dimension in  $\mat{R}$ except for dimension $i$ gives $\mat{p_i}$ vector. Then we have,
\begin{align*}
h(\mat{R}) &= -\sum_{i_1,i_2,\hdots,i_m}r_{i_1,i_2,\hdots,i_m}\log(r_{i_1,i_2,\hdots,i_m}) = -\sum_{i_1,i_2,\hdots,i_m}Ts_{i_1,i_2,\hdots,i_m}\log(Ts_{i_1,i_2,\hdots,i_m})\\
& = T\left(-\sum_{i_1,i_2,\hdots,i_m}s_{i_1,i_2,\hdots,i_m}\log(s_{i_1,i_2,\hdots,i_m})\right)-T\log(T)\\
& = T\left(\frac{1}{T}\Big(\sum_ih(\mat{p_i})+mT\log(T)\Big)\right) - T\log(T)\\
& = \sum_ih(\mat{p_i})+(m-1)T\log(T)
\end{align*}

Suppose $\mat{R}(i_1,i_2,\hdots,i_m) = \frac{\prod_jp_j(i_j)}{T^{m-1}}$. Then we have,
\begin{align*}
h(\mat{R} ) &= -\sum_{i_1,i_2,\hdots,i_m}\Big(\frac{\prod_jp_j(i_j)}{T^{m-1}}\Big)\log\Big(\frac{\prod_jp_j(i_j)}{T^{m-1}}\Big)\\
& = \frac{1}{T^{m-1}}\Big(T^{m-1}\sum_{i}h(\mat{p_i})+(m-1)T^m\log(T)\Big)=\sum_{i}h(\mat{p_i})+(m-1)T\log(T). 
\end{align*}

\end{proof}

Following Lemma \ref{lem:jointNonDistributionMultiple}, the maximum contribution of the second phase to the entropy is obtained when we place the scaled outer product of the remaining probability values on the joint probability matrix. The remaining probabilities after Phase 1 are $l_i$ for $X_i$ for all $i\in[m]$. The remaining probability mass is $\sum_il_j(i)=T,\forall j\in[m]$. Thus, in Phase II, $l_j,j \in[m]$ contributes the entropy of  $H_{Ph2}\leq \sum_jh(l_j) + (m-1)T\log{T} $. Finally, we can write
\begin{align}
H(U) &= H_{Ph1}+H_{Ph2} \leq  H_{Ph1}+\sum_jh(l_j) + (m-1)T\log{T}\nonumber\\
&\leq \sum_j\alpha_j H(X_j)+1-\sum_j\alpha_j h(l_j)+\sum_jh(l_j) + (m-1)T\log{T}\\
&\leq H^*(X_1,X_2,\hdots,X_m)+1-(m-1)T\log(1/T)+\sum_jh(l_j)-\max_jh(l_j).\label{eq:avgBoundMulti}
\end{align}
(\ref{eq:avgBoundMulti}) is obtained by selecting $\alpha_j=1$ for $j=\arg\max_k\{h(l_k)\}$, and through the bound 
\begin{equation}
H^*(X_1,X_2,\hdots,X_m)\geq \max(H(X_1),H(X_2),\hdots,H(X_m))\geq \sum_{j}\alpha_jH(X_j).
\end{equation} 
\end{proof}

\subsection{A special family of distributions}
\label{sec:example}
Let $X_1$ be uniformly distributed random variable over $n$ states, i.e., $\mu_1(i)=1/n,\forall i\in[n]$. Let $X_2$ have the distribution $\mu_2$ with the following: $\mu_2(i)=\frac{\alpha}{n},\forall i\in [n/2]$ and $\mu_2(i)=\frac{2-\alpha}{n},\forall i\in\{n/2+1,n/2+2,\hdots,n\}$, where $1< \alpha < 2$. One can check that $\mu_2$ sums to 1 with this parameterization. We can calculate the entropies of $X_1$ and $X_2$ which yields $H(X_1)=\log(n), H(X_2)=\log(n)-\frac{\alpha}{2}\log(\alpha)-\frac{2-\alpha}{2}\log(2-\alpha)$. Running Algorithm \ref{alg:alternative} on $X_1$ and $X_2$, we have the following:
\begin{align}
H(U) &= \log(n)-\frac{\alpha-1}{2}\log(\alpha-1)-\frac{2-\alpha}{2}\log(2-\alpha)\nonumber\\
&=H(X_2)+\frac{\alpha}{2}\log(\frac{\alpha}{\alpha-1})+\frac{1}{2}\log(\alpha-1)\nonumber\\
&=H(X_2)+\frac{1}{2}\log(1+\epsilon)+\frac{\epsilon}{2}\log(1+\frac{1}{\epsilon})\label{eq:exampleBound}\\
&\leq H(X_2	)+1\leq H^*(X_1,X_2)+1.
\end{align}
where in (\ref{eq:exampleBound}) we used the reparameterization $\alpha=\epsilon+1$ for $0 < \epsilon < 1$. Since $H(U)\geq H^*(X_1,X_2)$, algorithm outputs a joint distribution with entropy at most 1 bit away from the optimum. However, we have $h(l_1) = h(l_2)= \frac{\alpha-1}{2}\log\frac{n}{\alpha-1}$. Thus, $\min\{h(l_1),h(l_2)\} =\frac{\alpha-1}{2}\log\frac{n}{\alpha-1}$ yielding a gap of at least $ \frac{\alpha-1}{2}\log(n)$. In the light of this example, we believe that a tighter guarantee should be provable for the given algorithm.

\bibliographystyle{plain}
\bibliography{causalinferenceSimple.bib}

\end{document}